\documentclass[11pt]{article}
\usepackage[utf8]{inputenc}
\usepackage{amsmath,amssymb,amsbsy,amsfonts,amsthm,latexsym, amsopn,amstext,amsxtra,euscript,amscd,color,breqn,mathrsfs}

\PassOptionsToPackage{hyphens}{url}\usepackage{hyperref}

\usepackage[capbesideposition=outside,capbesidesep=quad]{floatrow}

\restylefloat{table}

\usepackage{multirow,caption}

\newfont{\teneufm}{eufm10}
\newfont{\seveneufm}{eufm7}
\newfont{\fiveeufm}{eufm5}

\usepackage{systeme}

\usepackage{xcolor}
\usepackage{tikz}
\usetikzlibrary{shapes}
\usetikzlibrary{decorations.markings}

\usepackage{pstricks}
\usetikzlibrary{shapes,arrows}
 \usepackage{pstricks-add}
\usepackage{epsfig}
 \usepackage[space]{grffile} 
 \usepackage{etoolbox} 
 \makeatletter 
 \patchcmd\Gread@eps{\@inputcheck#1 }{\@inputcheck"#1"\relax}{}{}
 \makeatother

\newtheorem{thm}{Theorem}[section]

\newtheorem{rem}[thm]{Remark}

\newtheorem{theorem}[thm]{Theorem}
\newtheorem{corollary}[thm]{Corollary}

\newtheorem{proposition}[thm]{Proposition}

\newtheorem{remark}[thm]{Remark}

\newcommand{\ff}[1]{{\mathbb F}_{#1}}
\newcommand{\ffs}[1]{{\mathbb F}_{#1}^\star}
\newcommand{\ffx}[1]{\ff{#1}[x]}

\newcommand{\ga}{\alpha}

\newcommand{\gd}{\delta}

\def\+{\oplus}

\def\F{{\mathbb F}}

\def\00{{\bf 0}}
\def\11{{\bf 1}}
\def\+{\oplus}

\def\\{\cr}
\def\({\left(}
\def\){\right)}

\newcommand{\cardinality}[1]{\# #1}

\usepackage[colorinlistoftodos,prependcaption,textsize=tiny]{todonotes}

\providecommand{\newoperator}[3]{%
  \newcommand*{#1}{\mathop{#2}#3}}

\newoperator{\FD}{\mathrm{FD}}{\nolimits}

\begin{document}
\title{\bf Low c-differential uniformity for functions modified on subfields}
\author{\Large Daniele Bartoli$^1$, Marco Calderini$^2$, Constanza Riera$^3$, \and \Large Pantelimon St\u anic\u a$^4$
\\ \\
$^1$ Department of Mathematics and Informatics, University of Perugia,\\
 Via Vanvitelli, 1, 06123, Perugia;\\
 {\tt daniele.bartoli@unipg.it}\\
 $^2$ Department of Informatics, University of Bergen\\
Postboks 7803, N-5020, Bergen, Norway;\\
 {\tt Marco.Calderini@uib.no}\\
$^3$Department of Computer Science,\\
 Electrical Engineering and Mathematical Sciences,\\
   Western Norway University of Applied Sciences,\\
  5020 Bergen, Norway; {\tt csr@hvl.no}\\
$^4$Applied Mathematics Department, \\
Naval Postgraduate School, \\
Monterey, CA 93943, USA; {\tt pstanica@nps.edu}
}
\maketitle

\begin{abstract}
In this paper, we construct some piecewise defined functions, and study their $c$-differential uniformity. As a by-product, we improve upon several prior results. Further, we look at concatenations of functions with low differential uniformity and show several results. For example, we prove that given $\beta_i$ (a basis of $\F_{q^n}$ over $\F_q$), some functions
 $f_i$ of $c$-differential uniformities $\delta_i$,  and $L_i$ (specific linearized polynomials defined in terms of $\beta_i$), $1\leq i\leq n$, then $F(x)=\sum_{i=1}^n\beta_i f_i(L_i(x))$ has $c$-differential uniformity equal to $\prod_{i=1}^n \delta_i$.
\end{abstract}


\section{Introduction and basic definitions}

Let $p$ be a prime number and $n$ be a positive integer. We let $\F_{p^n}$ be the  finite field with $p^n$ elements, and $\ffs{p^n}=\F_{p^n}\setminus\{0\}$ be its multiplicative group.

We call a function from $\F_{p^n}$ (or $\F_p^n$) to $\F_p$  a {\em $p$-ary  function} on $n$ variables. For positive integers $n$ and $m$, any map $F:\F_{p^n}\to\F_{p^m}$ (or  $\F_p^n\to\F_p^m$)  is called a {\em vectorial $p$-ary  function}, or an {\em $(n,m)$-function}. When $m=n$, $F$ can be uniquely represented as a univariate polynomial over $\F_{p^n}$
 of the form
$F(x)=\sum_{i=0}^{p^n-1} a_i x^i,\ a_i\in\F_{p^n}$,
whose {\em algebraic degree}   is then the largest weight  in the $p$-ary expansion of $i$ (that is, the sum of the digits of the exponents $i$) with $a_i\neq 0$.

  Motivated by~\cite{BCJW02}, who  extended the differential attack on some ciphers  by using a new type of differential,  in~\cite{EFRST20}, the authors introduced a new differential and Difference Distribution Table, in any characteristic, along with the corresponding perfect/almost perfect $c$-nonlinear functions and other notions (this was also developed independently in~\cite{BT} where the authors introduce the concept of quasi planarity). In~\cite{BC21,EFRST20,SPRS20,RS20}, various characterizations of the $c$-differential uniformity were found, and some of the known perfect and almost perfect nonlinear functions have been investigated.

For a $p$-ary $(n,m)$-function   $F:\F_{p^n}\to \F_{p^m}$, and $c\in\F_{p^m}$, the ({\em multiplicative}) {\em $c$-derivative} of $F$ with respect to~$a \in \F_{p^n}$ is the  function
\[
 _cD_{a}F(x) =  F(x + a)- cF(x), \mbox{ for  all }  x \in \F_{p^n}.
\]

For an $(n,n)$-function $F$, and $a,b\in\F_{p^n}$, we let the entries of the $c$-Difference Distribution Table ($c$-DDT) be defined by ${_c\Delta}_F(a,b)=\cardinality{\{x\in\F_{p^n} : F(x+a)-cF(x)=b\}}$. We call the quantity
\[
\delta_{F,c}=\max\left\{{_c\Delta}_F(a,b)\,:\, a,b\in \F_{p^n}, \text{ and } a\neq 0 \text{ if $c=1$} \right\},\]
the {\em $c$-differential uniformity} of~$F$. If $\delta_{F,c}=\delta$, then we say that $F$ is differentially $(c,\delta)$-uniform (or that $F$ has $c$-uniformity $\delta$). If $\delta=1$, then $F$ is called a {\em perfect $c$-nonlinear} ({\em PcN}) function (certainly, for $c=1$, they only exist for odd characteristic $p$; however, as proven in~\cite{EFRST20}, there exist PcN functions for $p=2$, for all  $c\neq1$). If $\delta=2$, then $F$ is called an {\em almost perfect $c$-nonlinear} ({\em APcN}) function.
It is easy to see that if $F$ is an $(n,n)$-function, that is, $F:\F_{p^n}\to\F_{p^n}$, then $F$ is PcN if and only if $_cD_a F$ is a permutation polynomial.

For $c=1$, we recover the classical derivative, PN, APN, differential uniformity and DDT.

In the last years, several constructions of low differentially uniform permutations have been introduced  by modifying some functions on a subfield (see for instance \cite{C21,PHT17,amc19,ZLS}).
In this work we will extend some of the results given in \cite{C21} to the case of the $c$-differential uniformity. From this generalization, we are also able to improve the upper bound obtained in \cite{S21} for the case of a Gold APN function in even characteristics.

\section{An upper bound on the differential uniformity of a piecewise defined function }

Here, we shall give a general result concerning an upper bound for the $c$-differential uniformity of a piecewise defined function, thus generalizing a result of~\cite{C21}.

Before considering the case of the $c$-differential uniformity, we will give a property for some functions having $\gd_{F,1}=4$ when $p=2$.
Indeed, recently in \cite{Carlet2021}, Carlet noticed that for an APN function $F\in\mathbb{F}_{2^s}[x]$ defined on an extension $\ff{2^{ms}}$, with $m$ odd, we have that the equation $F(x+a)+F(x)=b$ does not admit solutions $x\notin\ff{2^s}$, whenever $a\in\ffs{2^s}$ and $b\in\ff{2^s}$.
This result can be extended to the case of differentially $4$-uniform functions.

\begin{proposition}\label{prop:sub}
Let $n=sm$, with $m$ odd, and let $F\in\mathbb{F}_{2^s}[x]$ be a differentially $(1,4)$-uniform function over $\mathbb{F}_{2^n}$. Then, $F(x+a)+F(x)=b$ does not admit solutions $x\in\mathbb{F}_{2^n}\setminus\mathbb{F}_{2^s}$, whenever $a,b \in \mathbb{F}_{2^s}, a\neq0$.
\end{proposition}
\begin{proof}
Let us consider $a,b \in \mathbb{F}_{2^s}, a\neq0$. Let us denote by $x_1$, $x_1+a$, $x_2$, and $x_2+a$ the four solutions of $F(x+a)+F(x)=b$.

Suppose that one of this solutions is not in $\mathbb{F}_{2^s}$. Let us assume $x_1\notin\mathbb{F}_{2^s}$.

Note that the polynomial $F(x+a)+F(x)+b$ has all coefficients in $\mathbb{F}_{2^s}$, so if $x$ is a zero of the polynomial,  so is $x^{2^s}$. That means that $x_1^{2^s}$ is equal to $x_1$, $x_1+a$, $x_2$, or $x_2+a$.

Suppose $x_1^{2^s}=x_1$. Then, $x_1\in\F_{2^s}$, reaching a contradiction.

Suppose $x_1^{2^s}=x_1+a$. Then, we have  $x_1^{2^{2s}}=x_1^{2^{s}}+a=x_1$, implying that $x_1\in\mathbb{F}_{2^{2s}}\cap\mathbb{F}_{2^n}=\mathbb{F}_{2^s}$ (since $m$ is odd), which gives us a contradiction.

Consider, now, the case $x_1^{2^s}=x_2$ ($x_1^{2^s}=x_2+a$ is similar).
Then, we can have 4  different cases.

\noindent{\em Case $1$}. $x_2^{2^s}=x_2$:
This would imply $x_1^{2^s}\in\mathbb{F}_{2^s}$, which is not possible.

\noindent{\em Case $2$}. $x_2^{2^s}=x_1$:
We would have $x_1^{2^{2s}}=x_1$ and thus $x_1 \in \mathbb{F}_{2^s}$, not possible.

\noindent{\em Case $3$}. $x_2^{2^s}=x_2+a$:
From this, we obtain $x_2^{2^{2s}}=x_2$ and thus $x_1^{2^s}=x_2\in \mathbb{F}_{2^s}$, which is not possible.

\noindent{\em Case $4$}. $x_2^{2^s}=x_1+a$:
In this case, we get that $x_2^{2^{2s}}=x_1^{2^s}+a=x_2+a$, so we have  $x_2^{2^{4s}}=x_2^{2^{2s}}+a=x_2$, implying $x_2\in \mathbb{F}_{2^s}$ which gives us a contradiction.
\end{proof}

In the same way,
we can prove the following generalization: 

\begin{proposition}\label{prop:subs2k}
Let $n=sm$, where $s$ and $m$ are integers, and let $F\in\mathbb{F}_{2^s}[x]$ be a differentially $(1,2k)$-uniform function over $\mathbb{F}_{2^n}$, with $k\geq2$. If $m$ is not divisible by any integer $2\leq t\leq k$, then $F(x+a)+F(x)=b$ does not admit solutions $x\in\mathbb{F}_{2^n}\setminus\mathbb{F}_{2^s}$, whenever $a,b \in \mathbb{F}_{2^s}, a\neq0$.
\end{proposition}
\begin{remark}
We restrict to $k\geq2$ for ease of notation with the constrain $2\leq t\leq k$, but the result is true for $k=1$ if $m$ is odd, as proven in \cite{Carlet2021}.
\end{remark}
\begin{proof}
Let us consider $a,b \in \mathbb{F}_{2^s}, a\neq0$. Without loss of generality, we can suppose that the equation $F(x)+F(x+a)=b$ admits $2k$ solutions, that can be denoted by $x_1,\ldots,x_k,x_1+a,\ldots,x_k+a$. Suppose $x_1\notin\mathbb{F}_{2^s}$ and consider the set $O_{x_1}=\{x_1^{2^{is}}\,:\, 0\le i\le m-1\}=\{x_1^{2^{is}}\,:\, 0\le i\le 2k-1\}$. This last equality holds since the polynomial $F(x)+F(x+a)+b$ has all coefficients in $\mathbb{F}_{2^s}$, if $x$ is a solution, then also $x^{2^s}$ is a solution.

Now, if $|O_{x_1}|\le k$, then there exists $0<i\le k$ such that $x_1^{2^{is}}=x_1$, implying $x_1\in\mathbb{F}_{2^s}$, which gives us a contradiction.

If $|O_{x_1}|> k$, consider $J=\{j\,:\, x_j,x_j+a\in O_{x_1}\}$. We have $J\neq\emptyset$, and there must exists $j\in J$ for which there exist $0<i\le k$ such that $x_j^{2^{is}}=x_j+a$.
Indeed,
consider the sequence
$$
x_1,x_1^{2^s},\ldots
$$
and suppose that for all the pairs $x_j,x_j+a$ in this sequence we cannot have $x_j^{2^{is}}=x_j+a$, for $i\le k$. Then, up to relabeling the solutions, we would have that the first $k$ elements of the sequence are
$$
x_1,x_2(=x_1^{2^s}),\ldots,x_k.
$$
Now, for the next element we need to have one among $x_1+a,\ldots,x_k+a$. So, we would obtain a pair $x_j,x_j+a$ for which there exists $i\le k$ such that $x_j^{2^{is}}=x_j+a$.
Therefore, $x_j^{2^{2is}}=x_j$ for some $i\le k$ and  so $x_j\in\mathbb{F}_{2^s}$, implying $x_1\in\mathbb{F}_{2^s}$, contradiction.
\end{proof}

From Proposition \ref{prop:subs2k}, we can simplify Theorem 4.1. from \cite{C21} for some dimensions.
\begin{theorem}
\label{Caldesim}
Let $n=sm$, where $s$ and $m$ are integers. Let $f$ and $g$ be two polynomials with coefficients in $\F_{2^s}$, that is, $f,g\in\F_{2^s}[x]$, and $g$ permuting $\F_{2^s}$. Suppose that $f$ is a $\delta_{f,1}$-uniform function over $\F_{2^s}$ and $g$ is a $\delta_{g,1}$-uniform function over $\F_{2^n}$, and $m$ is not divisible by any integer $2\leq t\leq k$, where $k=\frac{\delta_{g,1}}{2}$. Then, the function
$$F(x)=f(x)+(f(x)+g(x))(x^{2^s}+x)^{2^n-1}=\begin{cases}
f(x),&\mbox{ if $x\in\mathbb{F}_{2^s}$,}\\
g(x),&\mbox{ if $x\notin\mathbb{F}_{2^s}$,}
\end{cases}$$
is such that
$$_1\Delta_F(a,b)\leq\begin{cases}
\max\{\delta_{f,1},\delta_{g,1}\},&\mbox{ if $a\in\mathbb{F}_{2^s}$,}\\
\delta_{g,1}+2,&\mbox{ if $a\notin\mathbb{F}_{2^s}.$}
\end{cases}$$
\end{theorem}

From Theorem~\ref{Caldesim}, we have that all the results given in \cite{C21} for the differentially $4$-uniform Gold and Bracken-Leander functions can be extended to other functions, such as the differentially $4$-uniform Kasami function. Indeed, the assumption on the solutions of the derivatives of the modified function is needed for applying Theorem 4.1 in \cite{C21}.
In particular, we have the following.

\begin{theorem}
Let $n=sm$, with $s$ even such that $s/2$ and $m$ are odd. Let $k$ be such that $\gcd(k,n)=2$ and $f(x)=A_1\circ Inv\circ A_2(x)$, where $Inv(x)=x^{-1}$ and $A_1,A_2$ are affine permutations over $\mathbb{F}_{2^s}$. Then
$$F(x)=f(x)+(f(x)+x^{2^{2k}-2^k+1})(x^{2^s}+x)^{2^n-1}=\begin{cases}
f(x),&\mbox{ if $x\in\mathbb{F}_{2^s}$},\\
x^{2^{2k}-2^k+1},&\mbox{ if $x\notin\mathbb{F}_{2^s}$,}
\end{cases}$$
is a differentially $(1,6)$-uniform permutation over $\mathbb{F}_{2^n}$. Moreover, if $s>2$ then the algebraic degree of $F$ is $n-1$. Moreover, the nonlinearity of $F$ is at least $2^{n-1}-2^{\frac{s}{2}+1}-2^{\frac{n}{2}}$.
\end{theorem}
\begin{proof}
The proof follows in a similar way as in \cite[Theorem~4.2, Proposition~4.1]{C21}.
\end{proof}

Theorem 4.1 in \cite{C21} can be extended to the case of $p$-ary functions and $c\ne1$. In the following result, we do not request any condition on the solutions of the derivatives of our functions.
Furthermore, we shall consider piecing more than two functions, but we prefer to state the result for two functions separately since it is the usual subfield modification, and the general case will be more evident.
\begin{thm}
\label{thm:t2}
 Let $p$ is a prime, $n>2$ be an integer, $s$ be a divisor of $n$, $1\neq c\in\F_{p^n}$ fixed, and $F:\F_{p^n}\to \F_{p^n}$ be a $p$-ary $(n,n)$-function defined by
 \[
 F(x)=\begin{cases}
 f(x), &\text{ if } x\in\F_{p^s},\\
 g(x), &\text{ if } x\notin\F_{p^s},
 \end{cases}
 \]
 where $f$ is an $(s,s)$-function of $c'$-differential uniformity $\delta_{f,c'}$ (for all $c'$) and $g\in\F_{p^n}[x]$
 is an $(n,n)$-function of $c'$-differential uniformity $\delta_{g,c'}$ (for all $c'$). Then, the $c$-differential uniformity of $F$ is 
\[
\delta_{F,c} \leq \begin{cases}
&  \delta_{f,0}+\delta_{g,0}, \text{ if } c= 0,\\
&  \max \left\{\delta_{f,c_1}+\delta_{g,c},\delta_{g,c}+2p^s\delta_{g,0} \right\}, \text{ if } c\neq 0,
\end{cases}
\]
where $c=\sum_{i=1}^m c_ig_i$, with $c_i\in\F_{p^s}$ and $\{g_1=1,g_2,\ldots,g_m\}$ is a basis of the extension  $\F_{p^n}$ over $\F_{p^s}$.

More generally, let $t\geq 2$,  $k_i\,|\,k_{i+1}$, $1\leq i\leq t-1$, $k_t=n$, be a sequence of integer divisors, and $f_i$, $1\leq i\leq t$, be some  $(k_i,k_i)$-functions of $c'$-differential uniformity $\delta_{f_i,c'}$ (for all $c'$).
 Further, let $c\in\F_{p^n}$ be fixed, and $F_t:\F_{p^n}\to \F_{p^n}$ be a $p$-ary $(n,n)$-function defined by
 \begin{align*}
     F_t(x)=
     \begin{cases}
f_1(x), &\text{ if } x\in\F_{p^{k_1}},\\
f_2(x), &\text{ if } x\in\F_{p^{k_2}}\setminus \F_{p^{k_1}},\\
\cdots &\cdots\cdots\cdots\\
f_t(x), &\text{ if } x\in\F_{p^{k_t}}\setminus 
\F_{p^{k_{t-1}}}.
     \end{cases}
 \end{align*} 
 Then, the $c$-differential uniformity of $f$ is
\[
\delta_{F_t,c}\leq \delta_{f_n,c}+\sum_{i=1}^{t-1} \max\left\{\delta_{f_i,c^{(i)}}, 2p^{k_i} \sum_{j=1}^{t-i-1} \delta_{f_j,0} \right\},
\]
where $c^{(i)}$ are the projections of $c$ onto $\F_{p^{k_i}}$, via some bases of $\F_{p^n}$ over $\F_{p^{k_i}}$.
\end{thm}
NB: Note that, if $c\in\F_{p^s}$, we have $c=c_1$, and $\delta_{f,c_1}=\delta_{f,c}$.
\begin{proof}
We first observe that the polynomial representation of $F$ is $F(x)=f(x)+(g(x)-f(x))(x^{p^s}-x)^{p^n-1}$ (here, we consider the embedding of $f$ as an $(n,n)$-function, by taking $f(x)=0$ for $x\notin\F_{2^s}$). We consider the $c$-differential equation, $F(x+a)-cF(x)=b$, of $F$ at $(a,b)\in\F_{p^n}\times \F_{p^n}$,
\begin{equation}
\begin{split}
\label{eq:cdiffF}
&f(x+a)+(g(x+a)-f(x+a))\left(x^{p^s}-x+a^{p^s}-a \right)^{p^n-1}\\
&\qquad\qquad \qquad
-cf(x)-c(g(x)-f(x))(x^{p^s}-x)^{p^n-1}=b.
\end{split}
\end{equation}
If $c=0$, the equation is either $f(x+a)=b$, or $g(x+a)=b$, depending upon $x+a$ being in $\F_{p^s}$ or not. The first claim follows.

If $c\neq 0$, we consider several cases.

\noindent
{\em Case $1$.} Let $a\in\F_{p^s}$. If $x\in\F_{p^s}$, Equation~\eqref{eq:cdiffF} becomes
\[
f(x+a)-cf(x)=b.
\]

Since $\F_{p^n}$ is an extension of degree $m$ over $\F_{p^s}$, we can write $c=\sum_{i=1}^m c_ig_i$ and $b=\sum_{i=1}^m b_ig_i$, where $b_i,c_i\in\F_{2^s}$ and $\{g_1=1,g_2,\ldots,g_m\}$ is a basis of the extension. Then, the equation above becomes
$$f(x+a)-\left(\sum_{i=1}^m c_ig_i\right)f(x)=\sum_{i=1}^m b_ig_i,
$$
which implies
$$f(x+a)-c_1f(x)=b_1 \mbox{ and } c_if(x)=b_i\ \forall i=2,\ldots,m.$$
This gives a (probably loose, though the $b_i$, and therefore the $\frac{b_i}{c_i}$, go through all values) bound for the number of solutions given by $\delta_{f,c_1}$.

NB: Note that, if $c\in\F_{p^s}$, we have $c=c_1$, and  this bound becomes $\delta_{f,c}$.

If  $x\notin\F_{p^s}$, Equation~\eqref{eq:cdiffF} transforms into
\[
g(x+a)-cg(x)=b,
\]
which has at most $\delta_{g,c}$ solutions.
Therefore, in this case we get at most $\delta_{f,c_1}+\delta_{g,c}$ solutions for~\eqref{eq:cdiffF}.

\noindent
{\em Case $2$.} Let $a\notin\F_{p^s}$. If $x+a\in\F_{p^s},x\notin\F_{p^s}$, then
Equation~\eqref{eq:cdiffF}  becomes
\begin{equation}
\label{eq:fg}
f(x+a)-cg(x)=b.
\end{equation}
We raise Equation~\eqref{eq:fg} to the power $p^s$ and get (using the fact that $(f(x+a))^{p^s}=f(x+a)$, since $x+a\in\F_{p^s}$ and $f$ is an $(s,s)$-function),
$f(x+a)-c^{p^s} g(x)^{p^s}=b^{p^s}$, which combined with~\eqref{eq:fg} renders
\begin{equation}
\label{eq:gg}
g(x)-c^{p^s-1} g(x)^{p^s}=\frac{b^{p^s}-b}{c}.
\end{equation}
The polynomial $c^{p^s-1}  X^{p^s}-X$ is a linearized
polynomial whose kernel is of dimension $s$. Thus, there are at most $p^s\delta_{g,0}$ (since for any root $X_0$ of $c^{p^s-1}  X^{p^s}-X+\frac{b^{p^s}-b}{c}$, there are at most $\delta_{g,0}$ values of $x$ such that $g(x)=X_0$) solutions to Equation~\eqref{eq:gg}.

Next, if $x+a\notin\F_{p^s},x\in\F_{p^s}$,  then \eqref{eq:cdiffF} becomes
$g(x+a)-cf(x)=b$, and an argument similar to the one above gives
\[
g(x+a)^{p^s}-c^{p^s-1} g(x+a)=b^{p^s}-c^{p^s-1}b,
\]
with at most $p^s\delta_{g,0}$ solutions.

It remains to consider $x,x+a\notin\F_{p^s}$. In this case, Equation~\eqref{eq:cdiffF} transforms into $g(x+a)-cg(x)=b$, which has at most $\delta_{g,c}$ solutions.
Putting these counts together, we obtain the first claim of the theorem.

For the general case, we use induction on $t$. The case of $t=2$ was treated in the first part of the proof, and the general case follows similarly.

If $c=0$, the same argument as before will show that $\delta_{F_t,0}\leq \sum_{i=1}^t \delta_{f_i,0}$. Using the notation
\begin{align*}
     F_{t-i+1}(x)=
     \begin{cases}
f_i(x), &\text{ if } x\in\F_{p^{k_i}},\\
f_{i+1}(x), &\text{ if } x\in\F_{p^{k_{i+1}}}\setminus\F_{p^{k_i}},\\
\cdots &\cdots\cdots\cdots\\
f_t(x), &\text{ if } x\in\F_{p^{k_t}}\setminus
\F_{t-1},
     \end{cases}
 \end{align*}
 and applying the induction assumption, we find that
 \[
 \delta_{F_t,c}\leq \delta_{F_{t-1},c}+ \max\{\delta_{f_1,c^{(1)}},2p^{k_1} \delta_{F_{t-1},0} \},
 \]
 if $c\neq 0$. By the first part of the proof, $\delta_{F_{t-1},0}\leq \sum_{i=1}^{t-1} \delta_{f_i,0}$. Moreover, $\delta_{F_{t-1},c}\leq \delta_{F_{t-2},c}+ \max\{\delta_{f_2,c^{(2)}}, 2p^{k_2} \sum_{i=1}^{t-2} \delta_{f_i,0} \}$,
 and by iteration we see that
 \[
 \delta_{F_t,c}\leq \delta_{f_n,c}+\sum_{i=1}^{t-1} \max\left\{\delta_{f_i,c^{(i)}}, 2p^{k_i} \sum_{j=1}^{t-i-1} \delta_{f_j,0} \right\}.
 \]
 The proof is done.
\end{proof}
\begin{rem}\label{reduce}
In the proof above, if $g\in\ffx{p^s}$, when $a\notin\ff{p^s}$ we can get:
for the case $x+a \in\ff{p^s}$
at most $\gd_{g,1/c^{p^s-1}}=\gd_{g,c^{p^s-1}}$ solutions;
and for the case $x \in\ff{p^s}$, we get
at most $\delta_{g,c^{p^s-1}}$ solutions.
Indeed, from Equation \eqref{eq:fg} we would have (recalling that $x+a\in \mathbb{F}_{p^s}$)
$$
g(x)^{p^s}-\frac{1}{c^{p^s-1}}g(x)=g(x+a-a^{p^s})-\frac{1}{c^{p^s-1}}g(x)=\frac{b-b^{p^s}}{c^{p^s}}.
$$
The number of solutions $x\notin\mathbb{F}_{p^s}$ such that $x+a\in\mathbb{F}_{p^s}$ is upper bounded by $\delta_{g,1/c^{p^s-1}}=\delta_{g,c^{p^s-1}}$.
The same for the case $x\in\mathbb{F}_{p^s}$ and $x+a\notin\mathbb{F}_{p^s}$. So, we have that for $c\ne 0$,  $\delta_{F,c}\le\max \left\{\delta_{f,c_1}+\delta_{g,c},\delta_{g,c}+2\delta_{g,c^{p^s-1}} \right\}$.
\end{rem}


Surely, there are other ways of piecing a function together, and we look at such a way below.
\begin{thm} 
 Let $p$ is a prime, $n>2$ be an integer, $n=st$, and $\gcd(s,t)=1$. Let $1\neq c\in\F_{p^n}$ fixed, and $F:\F_{p^n}\to \F_{p^n}$ be a $p$-ary $(n,n)$-function defined by
 \[
 F(x)=\begin{cases}
 f(x), &\text{ if } x\in\F_{p^t},\\
 g(x), &\text{ if } x\in\F_{p^s}\setminus\F_{p^t},\\
 h(x), &\text{ if } x\notin(\F_{p^s}\cup\F_{p^t}),
 \end{cases}
 \]
 where $f$ is a $(t,t)$-function of $c'$-differential uniformity $\delta_{f,c'}$ (for all $c'$),  $g$ is an $(s,s)$-function of $c'$-differential uniformity $\delta_{g,c'}$ (for all $c'$),
  and $h$ is an $(n,n)$-function of $c'$-differential uniformity $\delta_{h,c'}$ (for all $c'$). Then, the $c$-differential uniformity of $F$ is
 \[
\delta_{F,c} \leq \begin{cases}
&  \delta_{f,0}+\delta_{g,0}+\delta_{g,0}, \text{ if } c= 0,\\
&  \max \left\{\delta_{f,c_1}+\delta_{g,c_1'}+\delta_{h,c},\delta_{f,c_1}+2p^s\delta_{h,0}+\delta_{h,c},1+p^t\delta_{h,0}+\right.\\
&\left.\min\{p^t\delta_{g,0},p^s\delta_{f,0}\}+\delta_{g,c_1'}+\delta_{h,c},(2p^t+2p^s)\delta_{h,0}+\delta_{h,c} \right\}, \text{ if } c\neq 0,
\end{cases}
\]
where $c=\sum_{i=1}^s c_ig_i=\sum_{i=1}^t c_i'g_i'$, with $c_i\in\F_{2^t},c_i'\in\F_{2^s}$, and $\{g_1=1,g_2,\ldots,g_s\}$,  $\{g_1'=1,g_2',\ldots,g_t'\}$ are  bases of the extension $\F_{p^n}$ over $\F_{p^t}$, respectively, $\F_{p^n}$ over $\F_{p^s}$.
 \end{thm}
 \begin{proof}
 We need to investigate the number of solutions of $$F(x+a)-cF(x)=b.$$
 If $c=0$, for any $a,b$, the equation is either $f(x+a)=b$, $g(x+a)=b$ or $h(x+a)=b$. The first claim follows.

 Let now $c\neq0$ and $a=0$. In this case, $F(x+a)-cF(x)=(1-c)F(x)$, and we distinguish three cases:

 \noindent
 {\em Case $1$}. $x\in\F_{p^t}$: In this case, the equation is  $$(1-c)f(x)=b.$$
 As in the proof of Theorem 2.4, this implies that the number of solutions is upper bounded by $\delta_{f,c_1}$, where $c=\sum_{i=1}^s c_ig_i$, where $c_i\in\F_{2^t}$ and $\{g_1=1,g_2,\ldots,g_s\}$ is a basis of the extension of $\F_{p^n}$ over $\F_{p^t}$.

 \noindent
{\em Case $2$}. $x\in\F_{p^s}\setminus\F_p$: In this case, the equation is  $$(1-c)g(x)=b.$$ Similarly as in case 1), the number of solutions is upper bounded by $\delta_{g,c_1'}$, where $c=\sum_{i=1}^t c_i'g_i'$, where $c_i'\in\F_{2^s}$ and $\{g_1'=1,g_2',\ldots,g_t'\}$ is a basis of the extension of $\F_{p^n}$ over $\F_{p^s}$.

 \noindent
 {\em Case $3$}. $x\in\F_{p^n}\setminus(\F_{p^t}\cup\F_{p^s})$: In this case, the equation is  $$(1-c)h(x)=b.$$ The upper bound is here $\delta_{h,c}$.

 Let now $c\neq0$ and $a\in\F_{p^t}^*$. We can distinguish four cases:

 \noindent
 {\em Case $1$}.
 $x\in\F_{p^t}$, $x+a\in\F_{p^t}$: In this case the equation is $$f(x+a)-cf(x)=b.$$ As in the proof of Theorem 2.4, this implies that the number of solutions is upper bounded by $\delta_{f,c_1}$, where $c=\sum_{i=1}^s c_ig_i$, where $c_i\in\F_{2^t}$ and $\{g_1=1,g_2,\ldots,g_s\}$ is a basis of the extension of $\F_{p^n}$ over $\F_{p^t}$.

 \noindent
 {\em Case $2$}. $x\in\F_{p^s}\setminus\F_p$, $x+a\in\F_{p^n}\setminus(\F_{p^t}\cup\F_{p^s})$:
 In this case, the equation is
 $$h(x+a)-cg(x)=b.$$

 Raising to the power  $p^s$ and subtracting, we obtain the equation $$(h(x+a))^{p^s}-c^{p^s-1}h(x)=b^{p^s}-c^{p^s-1}b,$$
 which has as a solution set $b+\F_{p^s}$ (note that, if $c\in\F_{p^s}$, $c^{p^s-1}=1$, and, if $b\in\F_{p^s}$, $b+\F_{p^s}=\F_{p^s}$, so this covers all cases (with nonzero $c$)). The number of solutions is thus upper-bounded by $p^s\delta_{h,0}$.

  \noindent
 {\em Case $3$}. $x\in\F_{p^n}\setminus(\F_{p^t}\cup\F_{p^s})$, $x+a\in\F_{p^s}\setminus\F_p$:
 In this case, the equation is
 $$g(x+a)-ch(x)=b.$$ By similar arguments as the previous case, we obtain the bound $p^s\delta_{h,0}$.

 {\em Case $4$}. $x, x+a\in\F_{p^n}\setminus(\F_{p^t}\cup\F_{p^s})$: In this case, the equation is $$h(x+a)-ch(x)=b,$$
 so we have at most $\delta_{h,c}$.

 Let now $c\neq0,a\in\F_{p^s}\setminus\F_p$. We have now five cases:

  \noindent
 {\em Case $1$}. $x=0$, $x+a\in\F_{p^s}\setminus\F_p$. In this case, the equation is
 $$g(a)-cf(0)=b,$$
 which will be true for some $b$.

  \noindent
 {\em Case $2$}. $x\in\F_{p^t}^*$, $x+a\in\F_{p^n}\setminus(\F_{p^t}\cup\F_{p^s})$: In this case, the equation is $$h(x+a)-cf(x)=b,$$
 so we have at most $p^t\delta_{h,0}$ solutions.

  \noindent
 {\em Case $3$}. $x\in\F_{p^s}\setminus\F_p$, $x+a\in\F_{p^t}$: this case is only possible if $x+a\in\F_p$. Here the equation is $$f(x+a)-cg(x)=b.$$ If we raise to $p^t$, we see that the number of solutions is upper-bounded by $p^t\delta_{g,0}$. However, raising to $p^s$, we obtain an upper bound of $p^s\delta_{f,0}$. From this case, then, we get min$\{p^t\delta_{g,0},p^s\delta_{f,0}\}$.

  \noindent
 {\em Case $4$}. $x, x+a\in\F_{p^s}\setminus\F_p$: Here the equation is $$g(x+a)-cg(x)=b,$$ which gives an upper bound of $\delta_{g,c_1'}$ for the number of solutions,  where $c=\sum_{i=1}^t c_i'g_i'$, with $c_i'\in\F_{2^s}$ and $\{g_1'=1,g_2',\ldots,g_t'\}$ is a basis of the extension of $\F_{p^n}$ over $\F_{p^s}$.

  \noindent
 {\em Case $5$}. $x, x+a\in\F_{p^n}\setminus(\F_{p^t}\cup\F_{p^s})$. In this case, the equation is $$h(x+a)-ch(x)=b,$$ which gives an upper bound of $\delta_{h,c}$.

 Now, let us consider the case $c\neq0,a\in\F_{p^n}\setminus(\F_{p^t}\cup\F_{p^s})$. We consider here five cases:

  \noindent
 {\em Case $1$}. $x\in\F_{p^t}$, $x+a\in\F_{p^n}\setminus(\F_{p^t}\cup\F_{p^s})$. Here the equation is $$h(x+a)-cf(x)=b,$$
 which gives an upper bound of $p^t\delta_{h,0}$ for the number of solutions.

  \noindent
 {\em Case $2$}. $x\in\F_{p^s}\setminus\F_p$, $x+a\in\F_{p^n}\setminus(\F_{p^t}\cup\F_{p^s}).$ Here the equation is $$h(x+a)-cg(x)=b,$$
 which gives an upper bound of $p^s\delta_{h,0}$, for the count.

  \noindent
 {\em Case $3$}. $x\in\F_{p^n}\setminus(\F_{p^t}\cup\F_{p^s})$, $x+a\in\F_{p^t}$. Here the equation is $$f(x+a)-ch(x)=b,$$
 which gives an upper bound of $p^t\delta_{h,0}$, the number of solutions.

  \noindent
 {\em Case $4$}. $x\in\F_{p^n}\setminus(\F_{p^t}\cup\F_{p^s})$, $x+a\in\F_{p^s}\setminus\F_{p}$. Here the equation is $$g(x+a)-ch(x)=b,$$
 which gives an upper bound of $p^s\delta_{h,0}$, for the count of solutions.

  \noindent
 {\em Case $5$}. $x, x+a\in\F_{p^n}\setminus(\F_{p^t}\cup\F_{p^s})$. Here the equation is $$h(x+a)-ch(x)=b,$$
 which gives an upper bound of $\delta_{h,c}$.
 \end{proof}

 \begin{rem}
 Note that, if $c\in\F_{p^s}^*$ and $h\in\F_{p^s}[x]$ or $c\in\F_{p^t}^*$ and $h\in\F_{p^t}[x]$ we can reduce the bound, in a similar way as in Remark~\textup{\ref{reduce}}.
 \end{rem}

If we introduce some extra conditions on the solutions of the derivatives of the function $g$, we can obtain another upper bound on the $c$-differential uniformity of the modified function.

\begin{thm}\label{th:main}
 Let $p$ be a prime, $n>2$ be an integer, $s$ be a divisor of $n$, $1\neq c\in\F_{p^s}$ fixed, and $F:\F_{p^n}\to \F_{p^n}$ be a $p$-ary $(n,n)$-function defined by
 \[
 F(x)=\begin{cases}
 f(x), &\text{ if } x\in\F_{p^s},\\
 g(x), &\text{ if } x\notin\F_{p^s},
 \end{cases}
 \]
 where $f$ is an $(s,s)$-function of $c$-differential uniformity $\delta_{f,c}$ and $g\in\F_{p^s}[x]$ is an $(n,n)$-function of $c$-differential uniformity, $\delta_{g,c}$.
Suppose that:
\begin{itemize}
\item[$(H1)$] for any $a\in\ffs{p^s}$ and $b\in\ff{p^s}$ the equation $g(x+a)-g(x)=b$ has no solution in $\ff{p^n}\setminus\ff{p^s}$.
\item[$(H2)$] for any $a\in\ff{p^s}$ and $b\in\ff{p^s}$ the equation $g(x+a)-cg(x)=b$ has no solution in $\ff{p^n}\setminus\ff{p^s}$.
\end{itemize}
 Then, the $c$-differential uniformity of $F$ is
\[
_c\Delta_{F}(a,b) \leq \begin{cases}
\max\{\gd_{f,c},\gd_{g,c}\},&\mbox{ if $a\in\ff{p^s}$},\\
\gd_{g,c}+2\cdot\gd_{g,0},&\mbox{ if $a\notin\ff{p^s}$},
\end{cases}
\]
\end{thm}
\begin{proof}
In order to get the c-differential uniformity of $F$, we need to check the number of solutions of the equation
\begin{equation}\label{eq:derF}
F(x+a)-cF(x)=b.
\end{equation}
Let us consider $a\in\ff{p^s}$. Then, for a solution $x$, we can have that both $x$ and $x+a$ are in $\ff{p^s}$ or none is in $\ff{p^s}$. In the first case, \eqref{eq:derF} becomes
$$
f(x+a)-cf(x)=b,
$$
which has at most $\gd_{f,c}$ solutions if $b\in\ff{p^s}$ and none when $b\notin\ff{p^s}$.

In the second case, we obtain
$$
g(x+a)-cg(x)=b.
$$
From (H2) we have no solution in $\ff{p^n}\setminus\ff{p^s}$ if $b\in\ff{p^s}$.
If $b\notin\ff{p^s}$, the number of solutions is at most $\gd_{g,c}$.
Then, for $a\in\ff{p^s}$ we can have at most $\gd=\max\{\gd_{f,c},\gd_{g,c}\}$.

Let $a \notin \ff{p^s}$. Given a solution $x$ of \eqref{eq:derF}, we can have the following cases:
\begin{itemize}
\item[1.] $x\notin\ff{p^s}$ and $x+a\in\ff{p^s}$;
\item[2.] $x\in\ff{p^s}$ and $x+a\notin\ff{p^s}$;
\item[3.] both $x$ and $x+a$ are not in $\ff{p^s}$.
\end{itemize}
Let us consider Case 1. Then,  \eqref{eq:derF} becomes
\begin{equation}\label{eq:dermistaF}
f(x+a)-cg(x)=b.
\end{equation}
Let us note that $b\notin\ff{p^s}$, otherwise we cannot have a solution of this type since $g(x)\notin\ff{p^s}$, which is derived from (H2) with $a=0$.

From this, raising \eqref{eq:dermistaF} by $p^s$ and substracting \eqref{eq:dermistaF}, we obtain
$$g(x)^{p^s}-g(x)=-\left(\frac{b}{c}\right)^{p^s}+\frac{b}{c}.$$
Denoting by $y=g(x)$ and by $b'=-\frac{b}{c}$, we obtain
$$
y^{p^s}-y={b'}^{p^s}-b'.
$$

The solutions of this last equation are the elements of the coset $b'+\ff{p^s}$.
Now, $x\in a+\ff{p^s}$.
Therefore, we need to check how many elements we have in $g(a+\ff{p^s})\cap \left(b'+\ff{p^s}\right)$, where $g(a+\ff{p^s}):=\{g(x)\,:\,x\in a+\ff{p^s}\}$.
Suppose that $\left|g(a+\ff{p^s})\cap \left(b'+\ff{p^s}\right)\right|\ge 2$. Then, there exist $x_1,x_2,y_1,y_2\in\ff{p^s}$ such that $b+y_1=g(a+x_1)$, $b+y_2=g(a+x_2)$ and $x_1\ne x_2$, $y_1\ne y_2$. Thus,
$$g(a+x_1)-g(a+x_2)=y_1-y_2.$$ Denoting by $x'= a+x_2\notin\ff{p^s}$ and $a'=x_1-x_2\in\ff{p^s}$, we obtain that
$$g(x'+a')-g(x')=y_1-y_2.$$
This is not possible by (H1). Therefore, $\left|g(a+\ff{p^s})\cap \left(b'+\ff{p^s}\right)\right|\le 1$, implying that we have at most $\gd_{g,0}$ solutions in Case (1), since for any element $y$ in $g(a+\ff{p^s})$ we have $|g^{-1}(y)|\le\gd_{g,0}$.

For Case 2, we obtain, in a similar way, that $|g(a+\ff{p^s})\cap (b+\ff{p^s})|\le 1$, which implies that we have at most $\gd_{g,0}$ solutions.

In the last case, we obtain the equation
$$
g(x+a)-cg(x)=b,
$$
which admits at most $\gd_{g,c}$ solutions for any $b$. Then, for $a\notin\ff{p^s}$, Equation~\eqref{eq:derF} admits at most $\gd_{g,c}+2\cdot\gd_{g,0}$ solutions.
\end{proof}

\begin{remark}\label{rm:noh2}
We can note that if we remove condition $(H2)$ in Theorem \ref{th:main}, we would obtain that
$$
_c\Delta_{F}(a,b)\le\begin{cases}
\gd_{f,c}+\gd_{g,c}&\mbox{ if $a\in\ff{p^s}$}\\
\gd_{g,c}+2\cdot\gd_{g,0}&\mbox{ if $a\notin\ff{p^s}$}.
\end{cases}$$
Moreover, if $g$ permutes $\ff{p^s}$ then we have also that $\gd_{g,0}=1$.
\end{remark}

For PcN and APcN functions we have a similar result as in Proposition~\ref{prop:sub}.

\begin{proposition}\label{prop:subpcn}
Let $n=sm$, with $s$ and $m$ positive integers.  Let $c\in\ff{p^s}$ and $F\in\mathbb{F}_{p^s}[x]$. Then,
\begin{itemize}
    \item[i)] if $F$ is PcN, $F(x+a)-cF(x)=b$ does not admit solution $x\in\mathbb{F}_{p^n}\setminus\mathbb{F}_{p^s}$, whenever $a,b \in \mathbb{F}_{p^s}$.

    \item[ii)] if $F$ is APcN and $m$ is odd, $F(x+a)-cF(x)=b$ does not admit solution $x\in\mathbb{F}_{p^n}\setminus\mathbb{F}_{p^s}$, whenever $a,b \in \mathbb{F}_{p^s}$.
\end{itemize}
\end{proposition}
\begin{proof}
Suppose that F is APcN and $m$ is odd. We have then that the polynomial $F(x+a)-cF(x)-b$ admits at most $2$ roots for any $a$ and $b$. Then, if $a,b\in\ff{p^s}$, we have that if $x_1$ is a solution so is $x_1^{p^s}$, since $F(x+a)-cF(x)-b$ is a polynomial with coefficients over $\ff{p^s}$. 

Suppose next that $x_1\notin\ff{p^s}$. Then, $x_1^{p^s}=x_2$, where $x_2$ is the second root. So, $x_2^{p^s}$ must be equal to $x_1$, implying $x_1^{p^{2s}}=x_1$. Therefore $x_1\in \ff{p^{2s}}\cap\ff{p^n}=\ff{p^s}$, which gives us a contradiction.

For the PcN case, we have no restriction on $m$ since we have only one root $x_1$ of $F(x+a)-cF(x)-b$, and thus $x_1^{p^s}$ must be equal to $x_1$.
\end{proof}


As for the case of the differential uniformity we can extend the previous result as follows.
\begin{proposition}\label{prop:cdiffH2}
let $n=sm$, $c\in \ff{p^s}$, and $F$ is a $\gd_{F,c}$ $c$-differentially uniform function over $\ff{p^n}$, with coefficients on the subfield $\ff{p^s}$. Then, if for any prime $q\le \gd_{F,c}$, $q\nmid m$, the equation $F(x+a)-cF(x)=b$ does not admit solution $x\in\mathbb{F}_{p^n}\setminus\mathbb{F}_{p^s}$, whenever $a,b \in \mathbb{F}_{p^s}$.
\end{proposition}
\begin{proof}
Let $x$ be a solution of $F(x+a)-cF(x)=b$. Then, all the elements in $O_x=\{x^{p^{is}}\,:\, 0\le i\le m-1\}$ are solutions of this equation. Moreover, since $|O_x|\le \gd_{F,c}$, for some integer $j\le\gd_{F,c}$ we have $x^{p^{js}}=x$, implying that $x\in\ff{p^{\gcd(js,n)}}=\ff{p^s}$.
\end{proof}

\subsection{Shifting Gold-like functions on a subfield}

In~\textup{\cite{S21}}, the author studied the c-differential uniformity of the modified Gold function. In particular he obtained the following result.

\begin{theorem}\label{th:pante_gold}
Let $n=sm$. Let
$$
G(x)=x^{2^k+1}+\ga(x^{2^s}+x)^{2^{n}-1}+\ga=\begin{cases}
x^{2^k+1}+\ga, &\text{ if } x\in\F_{2^s},\\
x^{2^k+1}, &\text{ if } x\notin\F_{2^s},
 \end{cases}
$$
where $1\le k<n$, $\gcd(k,n)=1$, $\alpha \in\ffs{2^s}$. Then, for $c\ne 1$, the c-
differential uniformity of $G$ is $\delta_{G,c}\le 9$.
\end{theorem}

\begin{rem}
The c-differential uniformity of a Gold function $g(x)=x^{2^k+1}$ has been characterized in \textup{\cite[Theorem 4]{RS20}}. In particular, for $c\ne 1$ we have $\gd_{g,c}\le 2^{\gcd(k,n)}+1$. Applying Theorem~\textup{\ref{thm:t2}} and Remark~\textup{\ref{reduce}} we obtain that the $c$-differential uniformity of $G(x)=x^{2^k+1}+\ga(x^{2^s}+x)^{2^{n}-1}+\ga$ satisfies
$$
\delta_{G,c}\le \begin{cases}
2\cdot(2^{\gcd(k,n)}+1) & \text{ if $c=0$}\\
3\cdot(2^{\gcd(k,n)}+1) & \text{ if $c\ne0$}.
\end{cases}
$$
Therefore, the upper bound in Theorem~\textup{\ref{th:pante_gold}} can be obtained applying Theorem~\textup{\ref{thm:t2}} and Remark~\textup{\ref{reduce}}. Indeed, for $\gcd(k,n)=1$ we have
$$
\delta_{G,c}\le \begin{cases}
6 & \text{ if $c=0$}\\
9 & \text{ if $c\ne0$}.
\end{cases}
$$
\end{rem}

For a Gold-like function defined over $\ff{2^n}$, we can observe the following.
\begin{proposition}\label{prop:gold}
Let $n=sm$, with $m$ odd. For a Gold function $g(x)=x^{2^k+1}$ with $\gcd(n,k)=t$ such that $\ff{2^t}\subset \ff{2^s}$, we have that
$$
g(x+a)+g(x)=b
$$
does not admit solutions in $\ff{2^n}\setminus\ff{2^s}$, whenever $a\in\ffs{2^s}$ and $b \in\ff{2^s} $.
\end{proposition}
\begin{proof}
The proof follows in a similar way as Lemma 4.1 in \cite{C21}. Indeed, we can consider just the equation
$$
x^{2^k}+x=b.
$$
If $b\in\ff{2^s}$ we obtain that $(x^{2^k}+x)^{2^s}=x^{2^k}+x$, which implies $x^{2^s}+x\in\ff{2^k}$. Therefore, $x^{2^s}+x\in \ff{2^t}\subset\ff{2^s}$. Then, $(x^{2^s}+x)^{2^s}+x^{2^s}+x=0$ implies $x^{2^{2s}}=x$, and thus $x\in\ff{2^{2s}}\cap\ff{2^n}=\ff{2^s}$.
\end{proof}

\begin{rem}
Note that the above proposition cannot be derived directly from Proposition~\textup{\ref{prop:subs2k}}, for $t\geq2$. Indeed, the Gold-like function $g(x)=x^{2^k+1}$ with $\gcd(n,k)=t$ has differential uniformity equal to $2^t$. So, for applying Proposition~\textup{\ref{prop:subs2k}} we need $i\nmid m$ for any $2\le i\le 2^{t-1}$, while in Proposition~\textup{\ref{prop:gold}} we just require $2\nmid m$. For $t=1$, the result follows from~\textup{\cite{Carlet2021}}.
\end{rem}

\begin{theorem}
Let $n=sm$, with $m$ odd. For a Gold function $g(x)=x^{2^k+1}$, with $\gcd(n,k)=t$ such that $\ff{2^t}\subset \ff{2^s}$, and $n/t$ odd. Then, for any fixed $\ga\in\ffs{2^s}$,
$G(x)=x^{2^k+1}+\ga(x^{2^s}+x)^{2^{n}-1}+\ga$ is such that $\gd_{G,c}\le 3$, for any $c\in\ff{2^t}\setminus\{1\}$.
\end{theorem}
\begin{proof}
From Proposition \ref{prop:gold}, we have that $g(x)=x^{2^k+1}$ satisfies (H1) in Theorem \ref{th:main}.

Since $n/t$ is odd we have that $g$ is a permutation of $\ff{2^n}$, so $\gd_{g,0}=1$. Moreover, from Theorem 3.1 in \cite{BC21} we have that $g$ is PcN for $c\in\ff{2^t}\setminus\{1\}$.

From Proposition \ref{prop:subpcn} we have that (H2) holds.
Therefore, $\gd_{G,c}\le 3$ by Theorem \ref{th:main}.
\end{proof}

\begin{theorem}\label{th:gold}
Let $n=sm$, with $n$ odd. Given the Gold function $g(x)=x^{2^k+1}$ with $\gcd(n,k)=1$, then, for any fixed $\ga\in\ffs{2^s}$,
$G(x)=x^{2^k+1}+\ga+\ga(x^{2^s}+x)^{2^{n}-1}$ is such that $\gd_{G,c}\le 6$, for any $c\in\ff{2^s}\setminus\{1\}$.
Moreover, if $3\nmid m$, then $\gd_{G,c}\le 5$.
\end{theorem}
\begin{proof}
If $3\nmid m$, then since the map is $3$ $c$-differentially uniform from Proposition \ref{prop:cdiffH2} we have that (H2) in Theorem \ref{th:main} is satisfied. The same for (H1) by Proposition \ref{prop:sub}.
Therefore, from Theorem \ref{th:main} we have that $\gd_{G,c}\le 5$ ($\gd_{g,0}=1$).

If $3\mid m$, then we cannot guarantee that (H2) in Theorem \ref{th:main} is satisfied, but  applying Remark \ref{rm:noh2} we have $\gd_{G,c}\le 6$.
\end{proof}

\begin{remark}
Theorem~\textup{\ref{th:gold}} improves the upper bound obtained by St\u anic\u a in~\textup{\cite{S21}}, when $c$ is restricted to the subfield $\ff{2^s}$.
\end{remark}

\section{Concatenating functions with low c-differential uniformity}

In this section we will show how it is possible to obtain a function over $\ff{q^n}$, with low c-differential uniformity, concatenating $n$ functions defined over $\ff{q}$.

Let $\{\beta_1,\ldots ,\beta_n\}$ be a basis of $\mathbb{F}_{q^n}$ as vector space over $\mathbb{F}_q$. Let
$$
A=\left(
\begin{array}{cccc}
\beta_1 & \beta_1^{q}&\cdots &\beta_1^{q^{n-1}}\\
\beta_2 & \beta_2^{q}&\cdots &\beta_2^{q^{n-1}}\\
\vdots &&&\vdots\\
\beta_n & \beta_n^{q}&\cdots &\beta_n^{q^{n-1}}
\end{array}\right).
$$
The matrix $A$ is non-singular, so we let $A^{-1}=(a_{i,j})_{i,j}$.

Let us denote by $e_k$ the column vector composed by all zeros but one in position $k$, for $1\le k\le n$.
We define the linear polynomial
$$L_k(x)=\sum_{i=1}^{n} a_{i,k}x^{q^{i-1}}=(x,x^q,\ldots ,x^{q^{n-1}})\cdot A^{-1}\cdot e_k. $$

Any element $x\in \mathbb{F}_{q^n}$ can be written as $x=\beta_1 x_1+\cdots+\beta_n x_n$, with $x_i\in \mathbb{F}_q$. So, we have
$$
L_k(x)=\left(\sum_{i=1}^n \beta_ix_i,\ldots,\sum_{i=1}^n \beta_i^{q^{n-1}}x_i\right)\cdot A^{-1}\cdot e_k= (x_1,\ldots,x_n)\cdot A\cdot A^{-1}\cdot e_k=x_k.
$$
That is, $L_k$ is the projection of the $k$-th component of $x$.

So we obtain the following result.

\begin{theorem}\label{concatenation}
Let $c\in\mathbb{F}_q\setminus \{1\}$ and let $f_1,\ldots,f_n$ be $n$ functions over $\mathbb{F}_q$ with c-differential uniformity $\delta_1,\ldots,\delta_n$, respectively. Let $\beta_1,\ldots,\beta_n$, $L_k$ be defined as before. Then $F(x)=\sum_{k=1}^n\beta_k f_k(L_k(x))$ has c-differential uniformity equal to $\prod_{i=1}^n \delta_i$.
\end{theorem}
\begin{proof}
For any $a\in\mathbb{F}_{q^n}$, with $a=\beta_1a_1+\cdots+\beta_na_n$, we have
\begin{align*}
F(x+a)-cF(x)&=\sum_{k=1}^n\beta_k f_k(x_k+a_k)-c\sum_{k=1}^n\beta_k f_k(x_k)\\
&=\sum_{k=1}^n\beta_k (f_k(x_k+a_k)-cf_k(x_k)).
\end{align*}
So if we consider $b=\beta_1b_1+\cdots+\beta_nb_n$ we have
$$
F(x+a)-cF(x)=b, \text{ that is, }  f_k(x_k+a_k)-cf_k(x_k)=b_k, \text{ for all $k$}.
$$
The equation $f_k(x_k+a_k)-cf_k(x_k)=b_k$ admits at most $\delta_k$ solutions for any $a_k$ and $b_k$ in $\mathbb{F}_q$, and there exist some $a_k$ and $b_k$ for which we have $\delta_k$ solutions. So, we obtain that $F(x+a)-cF(x)=b$ admits at most $\prod_{i=1}^n \delta_i$ solutions and we can find $a$, and $b$ for which we obtain exactly $\prod_{i=1}^n \delta_i$ solutions.
\end{proof}

Using the previous result, we can construct a PcN function over $\mathbb{F}_{q^n}$ from~$n$ PcN functions over $\mathbb{F}_q$.
\begin{corollary}
Let $c\in\mathbb{F}_q\setminus \{1\}$ and let $f_1,\ldots ,f_n$ be $n$ functions over $\mathbb{F}_q$ that are PcN. Then $F(x)=\sum_{k=1}^n\beta_k f_k(L_k(x))$ is PcN.
\end{corollary}

\section{Concluding remarks}
\label{sec5}
In this work we extended some of the results given in \cite{C21} to the case of the $c$-differential uniformity. We piece together (in several ways) subfunctions and provide upper bounds for the $c$-differential uniformity of the obtained function. As a byproduct, we improve some prior results of~\cite{S21}. Further, we look at concatenations of functions with low differential uniformity and check how their $c$-differential uniformity changes. In particular, we prove that given $\beta_i$ (a basis of $\F_{q^n}$ over $\F_q$), some functions
 $f_i$ of $c$-differential uniformities $\delta_i$,  and some specific linearized polynomials $L_i$ defined in terms of $\beta_i$, $1\leq i\leq n$, then $\displaystyle F(x)=\sum_{i=1}^n\beta_i f_i(L_i(x))$ has $c$-differential uniformity equal to $\displaystyle \prod_{i=1}^n \delta_i$.
 We believe, it would be of interest to investigate these constructions for the case of the newly defined generalized boomerang uniformity, as in~\cite{S20} (see also~\cite{S20_Weil}, for other characterizations).

 \end{document}